\documentclass[draftclsnofoot, onecolumn, 12pt]{IEEEtran}

\IEEEoverridecommandlockouts

\usepackage{amsmath,graphicx}

\usepackage{floatflt}

\usepackage{color}
\usepackage{cite}

\usepackage{graphicx}
\usepackage{amsmath}
\usepackage{amsthm}
\usepackage{amssymb}

\usepackage{mathdots}

\usepackage{verbatim}

\usepackage{psfrag,array}

\usepackage{subfigure}

\usepackage{stfloats}

\usepackage{todonotes}

\usepackage{amsmath}
\usepackage{epstopdf}
\usepackage{algorithmic}

\newcommand{\p}[1]{\mathop{\mbox{\it p} } }

\renewcommand{\vec}[1]{\ensuremath{\boldsymbol{#1}}}
\newcommand{\be}{\begin{equation}}
\newcommand{\ee}{\end{equation}}
\newcommand{\ba}{\begin{array}}
\newcommand{\ea}{\end{array}}
\newcommand{\bea}{\begin{eqnarray}}
\newcommand{\eea}{\end{eqnarray}}
\newcommand{\bean}{\begin{eqnarray*}}
\newcommand{\eean}{\end{eqnarray*}}

\newcommand{\diag}{\mathop{\rm diag}}

\newcommand{\rmt}{^{\rm T}}

\definecolor{white}{rgb}{1,1,1}

\newtheorem{theorem}{Theorem}
\newtheorem{lemma}{Lemma}
\newtheorem{example}{Example}

\newtheorem{remark}{Remark}

\begin{document}

\title{An Identity of Hankel Matrices Generated from the Moments of Gaussian Distribution}

\author{\IEEEauthorblockN{ Sha Hu} \\
\IEEEauthorblockA{{Lund Research Center\\ Huawei Technologies Sweden AB, Lund, Sweden}\\
hu.sha@huawei.com}

}

\maketitle

\begin{abstract}
In this letter, we proved a matrix identity of Hankel matrices that seems unrevealed before, generated from the moments of Gaussian distributions. In particular, we derived the Cholesky decompositions of the Hankel matrices in closed-forms, and showed some interesting connections between them. The results have potential applications in such as optimizing a nonlinear (NL) distortion function that maximizes the receiving gain in wireless communication systems.
\end{abstract}


\section{Origin of the Problem}

Maximizing the receiving gain with a maximum-likelihood (ML) receiver in a nonlinear (NL) transmission system, leads to the seeking of an optimal distortion function. In~\cite{Rui2023}, the gain to be maximized with a distortion function $f(s)$, is defined as
\bea \label{G} G=\frac{\int_{-\infty}^{\infty}(f'(s))^2 p(s)\mathrm{d}s}{\int_{-\infty}^{\infty}f^2(s)p(s)\mathrm{d}s},\eea
where $p(s)$ is the probability distribution function (pdf) of an real-valued input signal $s$, with an average power $\sigma^2\!=\!E(s^2)$. Without loss of generality, we let $f(s)$ be an odd and polynomial function with $f(-s)\!=\!-f(s)$. Note that in~\cite{Rui2023} $G$ is defined for $\sigma\!=\!1$, while for general cases when the input and ouput signals are the same with $f(s)\!=\!s$, the gain equals $G\!=\!\sigma^{-2}$ and is not normalized. Nevertheless, we use the definition in (\ref{G}) for analysis.

The first-order derivative of $f(s)$ is thusly an even function and can be represented as
\bea \label{f's} f'(s)=\vec{a}\vec{D}\vec{z}\rmt, \eea
where $N\!=\!2M\!-\!1$ is the order of $f(s)$, and
\bea \vec{a}&\!\!\!=\!\!\!&(a_0, a_1, \cdots, a_{N-1}),  \\ 
 \label{z} \vec{z}&\!\!\!=\!\!\!&(1, s^2, \cdots,s^{N-1}), \\
 \label{d}  \vec{D}&\!\!\!=\!\!\!&\mathrm{diag}\!\left(1, 3, \cdots, N\right)\!. \eea
Here $\mathrm{diag}(\cdot)$ denotes the diagonal by puting the entries on its main diagonal. Hence, the function $f(s)$ reads
\bea \label{fs} f(s)=s\vec{a}\vec{z}\rmt.\eea
Note that if $\vec{a}$ is a solution that maximizes $G$ in (\ref{G}), so is any scaling of $\vec{a}$, and hence we also constrain $a_{N-1}\!=\!1$. 


Now, substituting (\ref{f's}) and  (\ref{fs}) into (\ref{G}) yields a classical Rayleigh quotient maximization with
\bea \label{newG} G=\frac{\vec{a}\vec{D}\left(\int_{-\infty}^{\infty}(\vec{z}\rmt\vec{z})p(s)\mathrm{d}s \right)\vec{D}\vec{a}\rmt}{\vec{a}\left(\int_{-\infty}^{\infty}(\vec{z}\rmt\vec{z})s^2p(s)\mathrm{d}s \right)\vec{a}\rmt}=\frac{\vec{a}\vec{D}\vec{A}\vec{D}\vec{a}\rmt}{\vec{a}\vec{B}\vec{a}\rmt},\eea
where the matrices $\vec{A}$ and $\vec{B}$ are Hankel matrices generated from the even moments of $s$, and equal to
\bea  \label{A}\vec{A}=\!\int_{-\infty}^{\infty}(\vec{z}\rmt\vec{z})p(s)\mathrm{d}s  
\!\!\!&=&\!\!\!\!\left(
\begin{matrix}
  1       & E(s^2) & E(s^4) & \cdots & E(s^{N-1}) \\
    E(s^2)      &E(s^4)  &E(s^6) & \cdots & E(s^{N+1}) \\
   \vdots& \iddots& \iddots& \iddots& \vdots \\
   E(s^{N-1})     &  E(s^{N+2})  & E(s^{N+3}) & \cdots &  E(s^{2(N-1)})
\end{matrix} \right)\!,  \\ 
\label{B} \vec{B}=\!\int_{-\infty}^{\infty}s^2(\vec{z}\rmt\vec{z})p(s)\mathrm{d}s  
\!\!\!&=&\!\!\!\!\left(
\;\begin{matrix}
  E(s^2)      & E(s^4) & E(s^6) & \cdots & E(s^{N+1}) \\
    E(s^4)      &E(s^6)  &E(s^8) & \cdots & E(s^{N+3}) \\
   \vdots& \iddots& \iddots& \iddots& \vdots \\
   E(s^{N+1})     &  E(s^{N+3})  & E(s^{N+5}) & \cdots &  E(s^{2N}) 
\end{matrix} \;\, \right)\!.  
\eea

Note that $\vec{B}$ can be obtained by shifting $\vec{A}$ by one row or one column, and by definitions both $\vec{A}$ and  $\vec{B}$ are positive semi-definite, and can therefore be factorized. Denoting the Cholesky decomposition of $\vec{B}$ as
\bea \label{chol} \vec{B}=\vec{L}\vec{L}\rmt,\eea
and letting
\bea  \label{al} \vec{x}=\vec{a}\vec{L},\eea
the Rayleigh quotient can be transformed into
\bea \label{newG1} G=\frac{\vec{x}\vec{L}^{-1}\vec{D}\vec{A}\vec{D}\vec{L}^{-T}\vec{x}\rmt}{\vec{x}\vec{x}\rmt}=\frac{\vec{x}\vec{C}\vec{x}\rmt}{\vec{x}\vec{x}\rmt},\eea
where
\bea \vec{C}=\vec{L}^{-1}\vec{D}\vec{A}\vec{D}\vec{L}^{-T}. \eea
Hence, the maximal gain $G$ equals the largest eigenvalue of $\vec{C}$, and the optimal $\vec{x}$ is the correspondent eigenvector.

\section{Main Results}

Although the problem can be general, our main results are derived with a Gaussian distribution $s\!\sim\!\mathcal{N}(0, \sigma^2)$ and $p(s)\!=\!\frac{1}{\sqrt{2\pi}\sigma}e^{\frac{-s^2}{2\sigma^2}}$. In which case, the moments are
\bea  E(s^n) =\left\{\begin{matrix} 0,& \text{if $n$ is odd}, \\ (n-1)!!\sigma^n,& \text{if $n$ is even}. \end{matrix} \right. \eea
The notation $n!!$ denotes the double factorial of $n$. 

Before proceeding further, we introduce some basics of Hermite polynomials~\cite{Hp}. The $n$th Hermite polynomial $H_n(s)$ is defined as
\bea \label{Hn}H_n(s)=(-1)^n e^{\frac{s^2}{2}}\frac{\mathrm{d}^n}{\mathrm{d} s^n}e^{-\frac{s^2}{2}}.\eea
For examples, the first six Hermite polynomials are
\bea \label{hpol}
  H_0(s)&\!\!\!=\!\!\!&1, \notag \\
    H_1(s)&\!\!\!=\!\!\!&s, \notag \\
      H_2(s)&\!\!\!=\!\!\!&s^2-1, \notag \\
        H_3(s)&\!\!\!=\!\!\!&s^3-3s, \notag \\
          H_4(s)&\!\!\!=\!\!\!&s^4-6s^2+3, \notag \\
            H_5(s)&\!\!\!=\!\!\!&s^5-10s^3+15s.
\eea 
The Hermite polynomials are orthogonal to each other, and with $s\!\sim\!\mathcal{N}(0, 1)$ it holds that
\bea \label{hnint}  
  \int_{-\infty}^{\infty}   H_n(s)  H_m(s)p(s)=n!\delta(n-m),\eea
where $n!$ is the factorial of $n$, and $\delta(n-m)$ is the Dirac delta-function with $\delta(0)\!=\!1$ and $\delta(m)\!=\!0$ for $m\!\neq\!0$. Further, the polynomials also follow the recursive formula
\bea \label{hndev} H'_n(s) =n H_{n-1}(s). \eea

We show an example with $N\!=\!5$ and $\sigma\!=\!1$.

\begin{example}
For $N\!=\!5$ ($M\!=\!3$), the matrices are equal to
\bea \vec{A}=\left(
\begin{matrix}
  1       & 1& 3 \\
1   &3 & 15 \\
3   &15 &  105
\end{matrix} \right)\!,  \quad
\vec{B}= \left(
\begin{matrix}
1     &3 & 15\\
3   &15 & 105 \\
15   &105  &  945\\
\end{matrix} \right)\!. 
\eea
It can be numerically shown that $\vec{C}\!=\vec{L}^{-1}\vec{D}\vec{A}\vec{D}\vec{L}^{-T}\!=\!\vec{D}\!=\!\mathrm{diag}(1, 3, 5)$. Hence, the maximum gain is $G\!=\!N\!=\!5$, and the optimal $\vec{x}\!=\![0, 0, 1]$ which yields $\vec{a}\!=\![15,  -10,  1]$. The optimal NL function is $f(s)\!=\!s^5\!-\!10s^3\!+\!15s$, which is identical to $H_5(s)$.
\end{example}

An interesting observation from this example is the identity ``$\vec{L}^{-1}\vec{D}\vec{A}\vec{D}\vec{L}^{-T}\!=\!\vec{D}$'' of Hankel matrices $\vec{A}$ and $\vec{B}$, which turns out to hold in general by further considering the case $\sigma\!\neq\!´1$, and is stated in the main theorem below.

\begin{theorem} 
The Hankel matrices $\vec{A}$ and $\vec{B}$, both of sizes $M\!\times\!M$ and generated from the even moments of a Gaussian distribution $\mathcal{N}(0, \sigma^2)$ as in (\ref{A}) and (\ref{B}), respectively, are equal to
\bea  \label{A2}\vec{A}\!\!\!\!&\!\!=\!\!&\!\!\! \vec{D}_\sigma\vec{A}_0 \vec{D}_\sigma, \\
\label{B2} \vec{B}\!\!\!\!&\!\!=\!\!\!\!&\!\!\! \sigma^2\vec{D}_\sigma\vec{B}_0 \vec{D}_\sigma,\eea
where the diagonal matrix $\vec{D}_\sigma$ is
\bea \vec{D}_\sigma\!=\!\diag(1, \sigma^2, \sigma^4, \cdots, \sigma^{N-1}). \eea
The Hankel matrices $\vec{A}_0$ and $\vec{B}_0$ are generated with the case $\sigma\!=\!1$ from (\ref{A}) and (\ref{B}), respectively, and equal to
\bea  \label{A1}\vec{A}_0\!\!&\!\!=\!\!&\! \left(
 \begin{matrix} 
  1       &1 & 3 & \cdots &(N-2)!! \\
    1      &3  &15 & \cdots &N!! \\
   \vdots& \iddots& \iddots& \iddots& \vdots \\
     (N-4)!!     &    (N-2)!!   &  N!!  & \cdots &   (2N-5)!! \\
  (N-2)!!     &    N!!   &  (N+2)!!  & \cdots &   (2N-3)!! 
\end{matrix} \right)\!\!, \notag \\   
 \eea
\bea
\label{B1} \vec{B}_0\!\!&\!\!=\!\!\!\!&\!\left(
 \begin{matrix}
   1      &3  &15 & \cdots &N!! \\
    3     &15  &105 & \cdots &(N+2)!! \\
   \vdots& \iddots& \iddots& \iddots& \vdots \\
  (N-2)!!     &    N!!   &  (N+2)!!  & \cdots &   (2N-3)!! \\
    N!!     &    (N+2)!!   &  (N+4)!!  & \cdots &   (2N-1)!! 
\end{matrix}  \right)\!\!.\notag \\ 
\eea
Denoting the Cholesky decomposition of $\vec{B}$ as in (\ref{chol}), and with the diagonal matrix $\vec{D}$ in (\ref{d}), it holds that
\bea \vec{L}^{-1}\vec{D}\vec{A}\vec{D}\vec{L}^{-T}=\sigma^{-2}\vec{D}.\eea
Or equivalently,
\bea \vec{D}\vec{A}\vec{D}=\sigma^{-2} \vec{L}\vec{D}\vec{L}^{T}. \eea
\end{theorem}

Before proving the theorem, we first show two lemmas. Lemma~1 provides the Cholesky decompositions of $\vec{A}$ and $\vec{B}$ in closed-forms, based on the Hermite polynomials. Note that in general to obtain the Cholesky decomposition of a Hankel matrix is not easy~\cite{P71, BV88}. Lemma~2 shows an interesting commuting property between the Cholesky decompositions of $\vec{A}$ and $\vec{B}$, and one can be derived easily from the other.

With $\vec{z}$ in (\ref{z}), we can rewrite the Hermite polynomials as
\bea \label{Hodd} H_{2m-1}(s)&\!\!\!=\!\!\!&s\vec{\beta}_{2m-1,N}\vec{z}\rmt, \\
H_{2m-2}(s)&\!\!\!=\!\!\!&\vec{\beta}_{2m-2,N}\vec{z}\rmt,
\eea
where $\vec{\beta}_{m,N}$ are row vectors of length $M$ ($N\!=\!2M\!-\!1$), with the last non-zero entry equal to 1. For instances, it holds from (\ref{hpol}) that
\bea 
\vec{\beta}_{0,N}&\!\!\!=\!\!\!&(1,0,\cdots,0), \notag \\
 \vec{\beta}_{1,N}&\!\!\!=\!\!\!&(1,0,\cdots,0), \notag \\
 \vec{\beta}_{2,N}&\!\!\!=\!\!\!&(-1,1,0\cdots,0), \notag \\
\vec{\beta}_{3,N}&\!\!\!=\!\!\!&(-3,1,0\cdots,0), \notag \\
\vec{\beta}_{4,N}&\!\!\!=\!\!\!&(3, -6, 1, 0,\cdots,0),\notag \\
\vec{\beta}_{5,N}&\!\!\!=\!\!\!&(15, -10, 1, 0,\cdots,0).
\eea 

\begin{lemma} 
Denoting two lower-triangular matrices,
\bea \label{Lb}
\vec{L}_b\!\!&\!\!\!=\!\!\!&\!\!\!\left( \vec{\beta}_{1,N}\rmt, \vec{\beta}_{3,N}\rmt,\cdots,\vec{\beta}\rmt_{N,N} \!\right)\rmt, \\
\vec{L}_a\!\!&\!\!\!=\!\!\!&\!\!\!\left( \vec{\beta}_{0,N}\rmt, \vec{\beta}_{2,N}\rmt,\cdots,\vec{\beta}\rmt_{N-1,N}\! \right)\rmt , \\
\eea
it holds that
\bea 
\label{eq1} \vec{L}_b\vec{D}_\sigma ^{-1}\vec{B}\vec{D}_\sigma ^{-1}\vec{L}\rmt_b&\!\!\!=\!\!\!&\sigma^2\vec{D}_b, \\
\label{eq2}\vec{L}_a\vec{D}_\sigma ^{-1}\vec{A}\vec{D}_\sigma ^{-1}\vec{L}\rmt_a&\!\!\!=\!\!\!&\vec{D}_a,
\eea
where the diagonal matrices are
\bea 
\vec{D}_a&\!\!\!=\!\!\!&\diag(1, 2!, 4!, \cdots, (N-1)!), \\
\vec{D}_b&\!\!\!=\!\!\!&\diag(1, 3!, 5!, \cdots, N!) \notag \\
&\!\!\!=\!\!\!&\vec{D}\vec{D}_a. \eea
\end{lemma} 

\begin{proof} 
From (\ref{A2}) and (\ref{B2}), it holds that
\bea \label{eq3}
\sigma^{-2}\vec{L}_b\vec{D}_\sigma ^{-1}\vec{B}\vec{D}_\sigma^{-1}\vec{L}\rmt_b \!\!\!&=&\!\!\!\vec{L}_b\vec{B}_0\vec{L}\rmt_b, \\
\label{eq4}  \vec{L}_a\vec{D}_\sigma ^{-1}\vec{A}\vec{D}_\sigma^{-1}\vec{L}\rmt_a\!\!\!&=&\!\!\!\vec{L}_a\vec{A}_0\vec{L}\rmt_a.
 \eea
Assuming $s\!\sim\!\mathcal{N}(0, 1)$, it holds that
\bea \label{eq5}
\vec{L}_b\vec{B}_0\vec{L}\rmt_b\!\!\!\!&=&\!\!\!\!\vec{L}_b\left(\int_{-\infty}^{\infty}(\vec{z}\rmt\vec{z})s^2p(s)\mathrm{d}s\right) \vec{L}\rmt_b \notag \\
\!\!\!\!&=&\!\!\!\!\int_{-\infty}^{\infty}(s\vec{L}_b\vec{z}\rmt)(s\vec{L}_b\vec{z}\rmt)\rmt p(s)\mathrm{d}s  \notag \\
\!\!\!\!&=&\!\!\!\!\int_{-\infty}^{\infty}\big(H_1(s), H_3(s),\cdots,H_{N}(s)\big)\rmt  \big(H_1(s), H_3(s),\cdots,H_{N}(s)\big) p(s)\mathrm{d}s 
\notag \\
\!\!\!\!&\!\!\!=\!\!\!&\!\!\!\vec{D}_b. 
\eea
The last equality holds from (\ref{hnint}). Similarly, it also holds that
\bea \label{eq6}
\vec{L}_a\vec{A}_0\vec{L}\rmt_a\!\!\!&\!\!\!=\!\!\!&\!\!\!\vec{L}_a\left(\int_{-\infty}^{\infty}(\vec{z}\rmt\vec{z})p(s)\mathrm{d}s \right)\vec{L}\rmt_a \notag \\
\!\!\!&\!\!\!=\!\!\!&\!\!\!\int_{-\infty}^{\infty}(\vec{L}_a\vec{z}\rmt)(\vec{L}_a\vec{z}\rmt)\rmt p(s)\mathrm{d}s  \notag \\
\!\!\!&\!\!\!=\!\!\!&\!\!\!\int_{-\infty}^{\infty}\big(H_0(s), H_2(s),\cdots,H_{N-1}(s)\big)\rmt   \big(H_0(s), H_2(s),\cdots,H_{N-1}(s)\big) p(s)\mathrm{d}s 
\notag \\
\!\!\!&\!\!\!=\!\!\!&\!\!\!\vec{D}_a. \eea
Hence, (\ref{eq1}) and (\ref{eq2}) follow from (\ref{eq3})-(\ref{eq6}).
\end{proof} 

\begin{lemma}
The following commuting property holds,
\bea \label{Lab} \vec{L}_b \vec{D}=\vec{D}\vec{L}_a. \eea
\end{lemma}
\begin{proof} 
Using the identity (\ref{hndev}), it holds that the first-order derivative of $\vec{L}_b\vec{z}$ satisfies
\bea 
(\vec{L}_b s\vec{z})'=\vec{L}_b \vec{D}\vec{z}=\vec{D}\vec{L}_a\vec{z},\eea 
which proves (\ref{Lab}).
\end{proof} 

Now we prove Theorem~1. Denoting $\vec{D}_b^{\frac{1}{2}}$ as the square-root of $\vec{D}_b$. From Lemma~1, the Cholesky decomposition of $\vec{B}$ satisfies 
\bea \vec{L}=\sigma\vec{D}_\sigma \vec{L}_b^{-1}\vec{D}_b^{\frac{1}{2}}.\eea
Hence, it holds that
\bea \label{main} \vec{L}^{-1}\vec{D}\vec{A}\vec{D}\vec{L}^{-T}\!\!\!&\!\!\!=\!\!\!&\!\! \sigma^{-2}(\vec{D}_b^{-\frac{1}{2}}\vec{L}_b\vec{D}_\sigma ^{-1})\vec{D}\vec{A}\vec{D}(\vec{D}_b^{-\frac{1}{2}}\vec{L}_b\vec{D}_\sigma ^{-1})\rmt \notag \\
\!\!\!&\!\!\!=\!\!\!&\!\!\sigma^{-2}\vec{D}_b^{-\frac{1}{2}}(\vec{L}_b\vec{D})\vec{A}_0(\vec{D} \vec{L}\rmt_b) \vec{D}_b^{-\frac{1}{2}} \notag \\
\!\!\!&\!\!\!=\!\!\!&\!\! \sigma^{-2}\vec{D}_b^{-\frac{1}{2}}\vec{D}(\vec{L}_a\vec{A}_0 \vec{L}\rmt_a)\vec{D} \vec{D}_b^{-\frac{1}{2}}  \notag \\
\!\!\!&\!\!\!=\!\!\!&\!\! \sigma^{-2}\vec{D}_b^{-\frac{1}{2}}\vec{D}\vec{D}_a\vec{D}\vec{D}_b^{-\frac{1}{2}}  \notag \\
\!\!\!&\!\!\!=\!\!\!&\!\!\sigma^{-2}\vec{D},\eea
which completes the proof.

\section{Some Remarks}

\begin{remark}
For the Rayleigh quotient (\ref{newG}) with $s\!\sim\!\mathcal{N}(0, \sigma^2)$, it holds that $\vec{C}\!=\!\sigma^{-2}\vec{D}\!=\!\sigma^{-2}\diag(1, 3, ...,N-1)$, and the maximal receiving gain is $G\!=\!\sigma^{-2}N$ which increases with the order of $f(s)$. Further, the optimal NL function is identical to the $N$th Hermite polynomial. 
\end{remark}

\begin{remark}
From Lemma~1, by noting that $\vec{L}_a$ and $\vec{L}_b$ are triangular matrices with all diagonal elements being 1, it holds that
\bea \sigma^{-2M^2}\cdot\det\vec{B}&\!\!\!=\!\!\!&1\cdot 3!\cdot 5!\cdot ~\cdots~\cdot N! ,\\
\sigma^{-2M(M-1)} \cdot\det\vec{A}&\!\!\!=\!\!\!&1\cdot 2!\cdot 4!\cdot ~\cdots~ \cdot (N-1)! , \\
 \sigma^{-2M}\cdot \frac{\det\vec{B}}{\det\vec{A}}&\!\!\!=\!\!\!&\det\!\vec{D}=1\cdot 3\cdot 5\cdot ~\cdots~ \cdot N . \quad \eea
\end{remark}
 
\begin{remark}
It also holds that
\bea \label{rec} \vec{A}\vec{D}+\vec{D}\vec{A}=\vec{A}+\sigma^{-2}\vec{B},\eea
or equivalently,
\bea  \sigma^{-2}\vec{B}= \vec{D}\vec{A}\vec{D}-(\vec{D}-\vec{I})\vec{A}(\vec{D}-\vec{I}). \eea
\end{remark}
\begin{proof}
Note that when $s\!\sim\!\mathcal{N}(0, \sigma^2)$, it holds that 
\bea p'(s)\!=\!-\sigma^{-2}s p(s),\eea
From the definitions of (\ref{fs}), (\ref{A}) and (\ref{B}), it then holds that
\bea {\label{recfor}} \vec{a}\vec{D}\vec{A}\vec{D}\vec{a}\rmt&\!\!\!=\!\!\!&\int_{-\infty}^{\infty}(f'(s))^2 p(s)\mathrm{d}s \notag \\ &\!\!\!=\!\!\!& \int_{-\infty}^{\infty}f'(s) p(s)\mathrm{d}f(s)    \notag \\ &\!\!\!=\!\!\!&f(s)f'(s) p(s)\bigg\vert_{-\infty}^{\infty}-  \int_{-\infty}^{\infty}f(s)(f''(s)p(s)+f'(s)p'(s))\mathrm{d}s \notag \\
&\!\!\!=\!\!\!& \int_{-\infty}^{\infty}(-f(s)f''(s)+\sigma^{-2} sf(s)f'(s))p(s))\mathrm{d}s \notag \\
&\!\!\!=\!\!\!& -\vec{a}\vec{A}(\vec{D}-\vec{I})\vec{D}\vec{a}\rmt+\sigma^{-2}\vec{a}\vec{B}\vec{D}\vec{a}\rmt, \eea
where we have used the facts that 
\bea sf''(s)\!=\!\vec{a}(\vec{D}-\vec{I})\vec{D}\vec{z}\rmt,\eea
and $f(\infty)f'(\infty)p(\infty)\!=\!f(-\infty)f'(-\infty) p(-\infty)\!=\!0$ when $f(s)$ is polynomial. This implies
\bea
\vec{D}\vec{A}\vec{D}=-\vec{A}(\vec{D}-\vec{I})\vec{D}+\sigma^{-2}\vec{B}\vec{D},\eea
which proves (\ref{rec}).
 \end{proof}

\bibliographystyle{IEEEtran}

\begin{thebibliography}{99}

\bibitem{Rui2023} D. Macedo, J. Guerreiro, R. Dinis, and S. Hu, ``On the design of nonlinear characteristics that optimize maximum likelihood OFDM performance,'' \textit{IEEE Trans. Veh. Technol.,} vol. 72, no. 12, pp. 16882-16886, Dec. 2023.


\bibitem{Hp} C. Hermite, ``Sur un nouveau développement en série de fonctions,'' \textit{C. R. Acad. Sci. Paris. 58: 93-100. Collected in Œuvres II}, pp. 293-303, 1864.

\bibitem{P71} J. Phillips, ``The triangular decomposition of Hankel matrices,'' \textit{Mathematics of Computation}, vo. 25, no. 115, pp. 599-602, Jul. 1971.

\bibitem{BV88} D. L. Boley, F. T. Luk, and D. Vandevoorde, ``A fast method to diagonalize a Hankel matrix,'' \textit{Linear Algebra and its Applications}, vol. 284, no. 1-3, pp. 41-52, Nov. 1988.



\end{thebibliography}

\end{document}